\documentclass[10pt, onecolumn]{IEEEtran}
\IEEEoverridecommandlockouts    

\usepackage[english]{babel}
\usepackage{enumitem}
\usepackage[latin1]{inputenc}
\usepackage{enumerate}
\usepackage{graphicx}
\usepackage{color}
\usepackage[dvipsnames]{xcolor}
\usepackage[T1]{fontenc}
\usepackage{subfigure}
\usepackage{dsfont}
\usepackage{amsfonts}
\usepackage{graphicx,wrapfig}
\usepackage[T1]{fontenc}
\usepackage{amsmath}
\usepackage{mathtools, cuted}
\usepackage{amsthm}
\usepackage{amstext}
\usepackage{amssymb}
\usepackage{mathrsfs}
\usepackage{cite}
\usepackage{mathtools}
\usepackage{tikz}
\usepackage{marginnote}
\usepackage{tkz-tab}
\usetikzlibrary{topaths,calc}
\usepackage{stackengine}

\def\R{\mathbb{R}}

\def\argmin{\mathop{\rm arg\, min}}
\def\eps{\varepsilon}

\def\A{{\mathcal A}}

\def\E{{\mathcal E}}

\def\C{{\mathcal C}}
\def\P{{\mathcal P}}

\usepackage[font=small,labelsep=space]{caption}
\DeclareCaptionLabelSeparator{dot}{.~}
\newcounter{example}
\newenvironment{example}[1][]{\refstepcounter{example}\par\medskip
	\noindent \textit{Example~\theexample. #1} \rmfamily}{\medskip}
\newtheorem{definition}{Definition}
\newtheorem{theorem}{Theorem}

\newtheorem{proposition}{Proposition}
\newtheorem{lemma}{Lemma}
\theoremstyle{remark}
\newtheorem{remark}{Remark}

\usepackage{times}
\usepackage{tikz}
\usepackage{amsmath}
\usepackage{verbatim}
\usetikzlibrary{arrows,shapes}
\tikzstyle{RectObject}=[rectangle,fill=white,draw,line width=0.2mm]
\tikzstyle{line}=[draw]
\tikzstyle{arrow}=[draw, -latex]
\usetikzlibrary{decorations.pathmorphing}
\usetikzlibrary{calc,shapes, positioning}
\usepackage{graphicx}
\usepackage{caption}
\usetikzlibrary{shapes.geometric}

\definecolor{DukeBlue}{HTML}{001A57}
\definecolor{DarkRed}{rgb}{0.75, 0.0, 0.0}
\definecolor{DarkGreen}{rgb}{0.0, 0.5, 0.0}

\IEEEoverridecommandlockouts

\allowdisplaybreaks

\title{\LARGE \bf
	Curvature of Hypergraphs via Multi-Marginal Optimal Transport
}

\author{Shahab Asoodeh$^{1}$, Tingran Gao$^{2}$, and James Evans$^{3}$
	\thanks{$^{1}$Computation Institute and Institute of Genomics and System Biology, The University of Chicago, Chicago, IL 60637 
		{\tt \small shahab@uchicago.edu}}
	\thanks{$^{2}$Department of Statistics,
		The University of Chicago, Chicago, IL 60637
		{\tt\small tingrangao@galton.uchicago.edu}}%
	\thanks{$^{3}$Computation Institute and Department of Sociology, The University of Chicago, Chicago, IL 60637
		{\tt\small jevans@uchicago.edu}}
}

\begin{document}

	\maketitle

	\begin{abstract}
		We introduce a novel definition of curvature for hypergraphs, a natural generalization of graphs, by introducing a multi-marginal optimal transport problem for a naturally defined random walk on the hypergraph. This curvature, termed \emph{coarse scalar curvature}, generalizes a recent definition of Ricci curvature for Markov chains on metric spaces by Ollivier [Journal of Functional Analysis 256 (2009) 810-864], and is related to the scalar curvature when the hypergraph arises naturally from a Riemannian manifold. We investigate basic properties of the coarse scalar curvature and obtain several bounds. Empirical experiments indicate that coarse scalar curvatures are capable of detecting ``bridges'' across connected components in hypergraphs, suggesting it is an appropriate generalization of curvature on simple graphs.
	\end{abstract}
	\section{Introduction}
	Complex systems or datasets are often modeled as weighted simple graphs in network science problems. While edges in these simple graphs \emph{qualitatively} characterize similarity or adjacency relations among entities represented by the graph vertices, the edge weights are frequently used to \emph{quantify} the nature of the interactions between pairs of nodes. Simple yet powerful as these simple graph models are, many recent work reported the importance of understanding \textit{higher-order interactions} among more than a pair of nodes, rendering simple graphs insufficient as a natural model for capturing the network structure information in these practices. Applications of this type include spatial network \cite{Hyper_Spatial}, image tagging \cite{Hyper_ImageTagging}, image retrieval \cite{Hyper_ImageRetrieval}, cellular networks \cite{HyperGraph_Cellular}, and co-authorship network \cite{Hyper_Co_Authorship}, to name just a few. \emph{Hypergraphs} have been proposed as a replacement to tackle this difficulty.

	Roughly speaking, a hypergraph $H=(V, \E)$ consists of a finite set $V$ of vertices and a set of \textit{hyperedges} $\E\subseteq 2^V$ -- just as edges in a simple graph that can be identified with vertex pairs, hyperedge $E\in \E$ are subsets of $V$. The ubiquitous influence in modeling complex networks fostered numerous recent developments in the theory and algorithms of hypergraphs, including extensive studies of the spectral and algebraic properties such as hypergraph Laplacian \cite{Hyper_Laplacian}, hypergraph partitioning \cite{Hyper_partitoining}, Cheeger's inequality for hypergraph \cite{Hyper_Cheegar}, and spectrum of hypergraphs\cite{Hyper_Spectrum}.
	
	Among many tools developed for better understanding the geometry of graphs, the \emph{graph Ricci curvature} \cite{Ollivier1,Ricci_Graph,weber,Bauer,Jost2014} has attracted an increasing amount of interest in the past years. In his original work \cite{Ollivier1}, Ollivier defined \emph{coarse Ricci curvature} for metric measure spaces, including simple graphs and Markov chains as special cases. In a nutshell, the coarse Ricci curvature summarizes the behavior of shortest paths with close-by starting points and parallel initial directions: two such paths tend to get closer to each other in a metric space of positive Ricci curvature, and further if the space is negatively curved. On simple graphs, this notion of Ricci curvature has found applications ranging from bounding the chromatic number \cite{OllivierRicci_Colouring_Paeng} and analyzing the Internet topology \cite{Internet_Ricco} to measuring the stability in financial networks \cite{Ricci_Risk},  brain structural connectivity \cite{Ricci_Brain}, and similarity of networks \cite{Ricci_Cancer}.
	
	This paper proposes a novel definition of curvature for hypergraphs by generalizing Ollivier's coarse Ricci curvature through a multi-marginal optimal transport framework (see e.g. \cite{CE2010,MultiMarginalOT4,MultiMarginalOT00}). Analogous to coarse Ricci curvatures, our definition of hypergraph curvature is grounded upon differential geometric intuitions, and reduces to the graph Ricci curvature when the hypergraph is a simple graph. While the coarse Ricci curvature is defined for pairs of points, which are naturally identified with edges in simple graphs, we need to adjust the construction to account for $\geq 3$ vertices joined by a hyperedge simultaneously; consequently, the geometric information captured in our definition is a summary of a small neighborhood enclosing all the end nodes of a hyperedge, as opposed to the directional information revealed by the Ricci curvature. In fact, our construction corresponds to the \emph{scalar curvature} of Riemannian manifolds under appropriate manifold assumptions analogous to the manifold Ricci curvature example in \cite{Ollivier1}. 

	
	
	

	The rest of the paper is structured as follows. We list a few useful notation in Section~\ref{sec:notations}.    After a brief review of coarse Ricci curvature in Section~\ref{sec:graph-ricci-curv}, we define \emph{coarse scalar curvature} in Section~\ref{sec:coarse-scal-curv}. Based on this definition, we then propose our notion of hypergraph curvature and investigate its properties in Section~\ref{sec:hypergraph-case}. We derive a closed form for the curvature of complete uniform hypergraphs in  Section~\ref{sec:Uniform_Hyper} and a general lower bound for \textit{hyperpaths} in Section~\ref{Sec:Example}. In Section~\ref{sec:coarse-scalar-curvature-manifold}, we provide a detailed consistency result for the definition of coarse scalar curvature in a Riemannian manifold setting. Finally, we conclude the paper in Section~\ref{sec:Conclusion}.

	\section{Hypergraph Curvature via Multi-Marginal Optimal Transport}
	
	\subsection{Notations}
	\label{sec:notations}
	For each vertex $i$ of a hypergraph $H=(V, \E)$, we use $d_i\coloneqq \sum_{E\in \E}1_{\{i\in E\}}$ to denote the degree of vertex $i$ and $d(E)\coloneqq \sum_{r\in V}1_{\{r\in E\}}$ to denote the cardinality of hyperedge $E\in \mathcal{E}$. Similar to graph, we use $N(i)$ for the neighbors of vertex $i$, i.e., $N(i)=\{j\in V: \exists E\in \E, (i, j)\in E\}$. For a pair of vertices $i$ and $j$ of a hypergraph, $d(i,j)$ denotes the shortest distance, i.e., $d(i, j)=r$ if there exist $r$ interesting hyperedges $E_1, \dots, E_r$, such that $i\in E_1$, $j\in E_r$, and $E_k\cap E_{k+1}\neq \emptyset$ for $1\leq k\leq r-1$.
	
	We will always denote $M$ for a $d$-dimensional Riemannian manifold.  We use $\exp_x(\cdot):T_xM\to M$ to denote the exponential map. For any $x\in M$ and $v\in T_xM$ with $\left\| v \right\|=1$, we denote $\mathrm{Ric}_x\left( v,v \right)$ for the \emph{Ricci curvature at $x\in M$ in the direction of $v\in T_xM$}, defined as
	\begin{equation}
	\label{eq:def-ricci}
	\mathrm{Ric}_x\left( v,v \right)=\frac{1}{d-1}\sum_{i=2}^d\left\langle R \left( v,e_i \right)v,e_i \right\rangle
	\end{equation}
	where $v,e_2,\cdots,e_d$ constitutes an orthonormal basis for $T_xM$, and $R$ is the Riemannian curvature tensor. We will often write $\mathrm{Ric}\left( v,v \right)$ for $\mathrm{Ric}_x \left( v,v \right)$ when the point $x$ is fixed throughout the discussion. Averaging out $\mathrm{Ric}_x \left( v,v \right)$ for $v$ ranging in an orthonormal basis of $T_xM$ gives rise to the \emph{scalar curvature at $x\in M$:}
	\begin{equation}
	\label{eq:def-scalar}
	\mathrm{Scal}\left( x \right)=\frac{1}{d}\sum_{i=1}^d\mathrm{Ric}_x\left( z_i,z_i \right)
	\end{equation}
	where $z_1,\cdots,z_d$ constitutes an orthonormal basis for $T_xM$. Equivalently, the scalar curvature can be obtained from averaging out the Ricci curvature over the unit sphere in the tangent plane, i.e. (c.f. \cite[Exercise 4.9]{doCarmo1992RG})
	\begin{equation}
	\label{eq:equiv-def-scalar}
	\mathrm{Scal}\left( x \right)=\frac{1}{\omega_{d-1}}\int_{S_{d-1}\left( 0 \right)}\mathrm{Ric}_x \left( v,v \right)\,\mathrm{d}S_{d-1}\left( v \right)
	\end{equation}
	where $S_{d-1}\left( 0 \right)$ is the unit sphere in $T_xM$ and $\omega_{d-1}$ is the volume of the standard $\left( d-1 \right)$-dimensional sphere in $\mathbb{R}^d$. 
	
	For any set $A\in \mathbb{R}^n$, we let $\mathscr{P}\left( A \right)$ denote the set of all probability measures defined on $A$.
	
	\subsection{Graph Ricci Curvature}
	\label{sec:graph-ricci-curv}
	
	Given a graph $G=(V, E)$ and a pair of vertices $x, y\in V$, Ollivier  \cite{Ollivier1} defined the curvature of edge $(x,y)\in E$ as
	\begin{equation}\label{Def_Ollivier}
	\kappa(x,y)\coloneqq 1-\frac{W(m_x, m_y)}{d(x,y)},
	\end{equation}
	where $d(x,y)$ is the shortest distance from $x$ to $y$, $m_i$ is the uniform random walk starting at $i\in V$, and $W(m_x, m_y)$ is the \textit{Wasserstein distance} between $m_x$ and $m_y$ given by 
	\begin{equation}\label{Def_Wasserstein}
	W(m_x, m_y)\coloneqq\inf_{\pi\in \Pi(m_x, m_y)}\sum_{(r, s)\in V^2}d(r, s)\pi(r,s),
	\end{equation}
	where $\Pi(m_x, m_y)$ is the set of all joint distributions having $m_x$ and $m_y$ as marginals (i.e., the set of all couplings of $m_x$ and $m_y$). 
	He then showed that  positive curvature is equivalent to the contraction of the random walk under Wasserstein's distance which in turn leads to the existence of a unique stationary distribution. 
	
	To justify that \eqref{Def_Ollivier} is a valid discrete version of Ricci curvature, Ollivier argued as follows. In a $d$-dimensional Riemannian manifold $(M,d_M)$,  consider the random walk (c.f. Definition~\ref{Definition_RandomWalk}) $\text{d}m_x^\eps$ for each $x\in M$ and  $\eps>0$ given by
	\begin{equation}
	\label{eq:random-walk-riemannian-manifold}
	\text{d}m_x^\eps(s)\coloneqq \frac{1_{B\left(x, \eps\right)}\left( s \right)}{\text{vol}(B(x, \eps))}\text{d}\text{vol}(s),
	\end{equation}
	where $B \left( x,\eps \right)=\left\{ z\in M\mid d_M \left( x,z \right)<\eps \right\}$ is the metric ball with radius $\eps$ centered at $x$ and $1_{\{x\in A\}}$ is the indicator function of set $A$. It is then shown \cite[Example 7]{Ollivier1} that for sufficiently small $\delta=d_M\left(x,y\right)>0$
	\begin{equation}\label{Manifold_Ricci}
	1-\frac{W_1 \left( m^{\eps}_x,m^{\eps}_y \right)}{d_M \left( x,y \right)} = \frac{\eps^2 \mathrm{Ric}(v,v)}{2(n+2)} + O(\eps^3 + \eps^2\delta),
	\end{equation}
	where $v\in T_x M$ is a unit tangent vector at $x$ such that $\exp_x(\delta v)=y$ and where $W_1 \left( m^{\eps}_x,m^{\eps}_y  \right)$ is the \emph{$L^1$-Wasserstein distance} between $m^{\eps}_x$ and $m^{\eps}_y$
	\begin{equation}\label{Wasserstein_L1}
	W_1 \left( m^{\eps}_x,m^{\eps}_y \right):=\inf_{\pi\in \Pi \left( m^{\eps}_x,m^{\eps}_y \right)}\int_{M\times M}d_M \left( x,y \right)\,\mathrm{d}\pi \left( x,y \right).
	\end{equation}
	
	In Riemannian geometry, Ricci curvature $\mathrm{Ric}_x(v,v)$ is, up to a scaling factor, the average of the sectional curvatures of all two-dimensional subspaces of $T_xM$ passing through $v$ \cite[\textsection 4]{doCarmo1992RG} and hence, it measures the coupling of the random walks.  In this context, if the curvature of a point  in a manifold is zero, it is locally on a Euclidean space, positive if it is locally on an sphere and negative if it locally on a hyperbolic space.
	
	\subsection{Coarse Scalar Curvature for Metric Spaces}
	\label{sec:coarse-scal-curv}
	We begin our construction of coarse scalar curvature by defining \emph{random walks} on a metric space, which was used in \cite{Ollivier1} to define coarse Ricci curvatures.
	\begin{definition}[\cite{Ollivier1} Definition 1] \label{Definition_RandomWalk}
		Let $\left( M,d_M \right)$ be a Polish metric space equipped with its Borel $\sigma$-algebra. A \emph{random walk} $m$ on $M$ is a family of probability measures $\left\{ m_x\mid x\in M \right\}$ satisfying
		\begin{itemize}
			\item The map $x\mapsto m_x$ is measurable;
			\item Each $m_x$ has finite first moment.
		\end{itemize}
	\end{definition}
	
	\begin{definition}[Coarse Scalar Curvature]\label{Def:Coarse_Scalar}
		For a collection of $n$ points $\mathscr{X}_n:=\left\{x_1,\cdots,x_n\right\}$ in a metric space $\left( M, d_M \right)$ with random walk $m:=\left\{ m_x\mid x\in M \right\}$, define the \emph{coarse scalar curvature} for $\mathscr{X}_n$ by
		\begin{equation}
		\label{eq:coarse-scalar-curvature}
		\kappa \left( \left\{ \mathscr{X}_n \right\} \right):=1-\frac{W_1 \left( \mathscr{X}_n \right)}{\displaystyle c(x_1, \dots, x_n)},
		\end{equation}
		where $W_1 \left( \mathscr{X}_n \right)$ is the minimum of the multi-marginal optimal transport problem
		\begin{equation*}
		W_1 \left( \mathscr{X}_n \right):=\inf_{\pi\in \Pi \left( \mathscr{X}_n \right)}\int_{M\times\cdots\times M}c \left( \xi_1,\cdots,\xi_n \right)\,\mathrm{d}\pi \left( \xi_1,\cdots,\xi_n \right),
		\end{equation*}
		with
		\begin{equation*}
		\Pi \left( \mathscr{X}_n \right):=\Pi \left( m_1,\cdots,m_n \right)\\
		=\left\{ \pi\in \mathscr{P}\left( M\times\cdots\times M \right)\mid\left(\mathrm{Proj}_k\right)_{\#}\pi =m_{x_k},\forall 1\leq k\leq n\right\},
		\end{equation*}
		with $(A)_{\#}\pi$ being the push-forward of measure $\pi$ under mapping $A$ and 
		\begin{equation}\label{Def:Cost}
		c \left( \xi_1,\cdots,\xi_n \right):=\inf_{z\in M}\sum_{i=1}^n d_M \left( \xi_i,z \right).
		\end{equation}
	\end{definition}
	
	It must be mentioned that the multi-marginal optimal transport problem is first studied in Gangbo and \'Swi\c ech \cite{MultiMarginalOT00} where they showed the necessary conditions for the existence and uniqueness of the optimizer when $c(\xi_1,\cdots,\xi_n)$ is the sum of pairwise $\ell^2$ distances. 
	The coarse scalar curvature is closely tied to the multi-marginal optimal transport problem among $n\geq 2$ probability measures $m_{x_1},\cdots,m_{x_n}$, which is a direct generalization of the pairwise Wasserstein distance between measures. We will make frequent use of the following well-known facts in the theory of multi-marginal optimal transport problems (see e.g. \cite{CE2010,MultiMarginalOT00,MultiMarginalOT4,Barycenter_Agueh,Kellerer1984} and the references therein):
	\begin{proposition}[Multi-marginal and barycenter \cite{CE2010}]\label{Prop:Equivalent_barycenter}
		The minimum of the multi-marginal optimal transport problem is equal to the minimum of the Wasserstein barycenter problem, i.e.
		\begin{equation}\label{Barycenter_Definition}
		W_1 \left( \mathscr{X}_n \right)=\inf_{\nu\in\mathscr{P}\left( M \right)}\sum_{i=1}^n W_1 \left( m_{x_i},\nu \right)
		\end{equation}
		where the $L^1$-Wasserstein distance $W_1 \left( m,\nu \right)$ for any $m,\nu\in\mathscr{P}\left( M \right)$ is defined in \eqref{Wasserstein_L1}.
	\end{proposition} 
	The minimization problem in \eqref{Barycenter_Definition} is called  Wasserstein barycenter of $\mathscr{X}_n$.
	
	\begin{proposition}[Duality \cite{Kellerer1984}]\label{Prop:Duality}
		\begin{equation}
		\label{eq:multi-marginal-duality}
		W_1 \left( \mathscr{X}_n \right)=\sup_{\left( f_1,\cdots,f_n \right)\in\mathscr{F}}\sum_{i=1}^n\int_Mf_i \left( \xi \right)\,\mathrm{d}m_{x_i} \left( \xi \right)
		\end{equation}
		where the supremum is taken over
		\begin{equation*}
		\mathscr{F}:=\Big\{ \left( f_1,\cdots,f_n \right) \,\Big|\, f_i\in L^1 \left( M,m_{x_i} \right)\,\,\forall 1\leq i\leq \left| e_x \right|,\\
		\sum_{i=1}^nf_i \left( \xi_i \right)\leq c \left( \xi_1,\cdots,\xi_n \right) \Big\}.
		\end{equation*}
	\end{proposition}
	We postpone the justification of our nomenclature of ``scalar curvature'' in Section~\ref{sec:coarse-scalar-curvature-manifold}, under an appropriate manifold setting. Briefly speaking, at least when the hypergraph arises from a Riemannian manifold in a natural geometric construction, the coarse scalar curvature \eqref{eq:coarse-scalar-curvature} is asymptotically lower bounded by the scalar curvature of the Riemannian manifold.
	
	
	\subsection{Hypergraph Curvature}
	\label{sec:hypergraph-case}
	Let a hypergraph $H=(V, \E)$ with $V=\{1, 2, \dots, N\}$ and a hyperedge $E=\{1, 2, \dots, n\}$ be given. Inspired by coarse scalar curvature, we wish to define the curvature for each hyperedge $E$ of $H$. To this goal, we first need to define random walk over hypergraphs. It is natural to define the (uniform) random walk started at vertex $i\in V$ on $H$ as the following: for each $j\in V$
	\begin{equation}\label{Def_RandomWalk}
	m_i(j)=\sum_{E\in \mathcal{E}:(i, j)\in E}\frac{1}{d_i}\frac{1}{d(E)-1},
	\end{equation}
	and hence we associate $E$ with $n$ probability measures $m_i\in \mathscr{P}(V)$, $1\leq i\leq n$. Replacing $\mathscr{X}_n$ with $\{1, 2, \dots, n\}$ in Definition~\ref{Def:Coarse_Scalar} and defining $c(x_1,\dots, x_n)=\min_{z\in V}\sum_{i=1}^nd(x_i, z)$, we can define a multi-marginal optimal transport problem associated with $E\in \E$ as 
	\begin{equation}\label{Def_Multi_Marginal}
	W(E)\coloneqq\min_{\pi\in \Pi(m_1, m_2, \dots, m_n)}\sum_{\mathbf{x}^n\in V^n} c(\mathbf{x}^n)\pi(\mathbf{x}^n),
	\end{equation}
	where $\mathbf{x}^n=(x_1,x_2, \dots, x_n)$.
	It is worth mentioning that $c(1, 2, \dots, n)=n-1$ as $\{1, 2, \dots, n\}$ is a hyperedge. In general, $c(x_1, \dots, x_n)\geq n-1$ for distinct $x_1, \dots, x_n$. 
	\begin{definition}[Hypergraph Curvature]\label{Def_Hypergraph_Curva}
		The curvature of a hyperedge $E=\{1, 2, \dots, n\}\in \E$ is defined as 
		$$\kappa(E)\coloneqq 1-\frac{W(E)}{n-1}.$$
	\end{definition}
	Note that if $x_i\notin N(i)$, then $\pi(x_1, \dots,x_i, \dots, x_n)=0$ for all $\pi\in\Pi(m_1, \dots, m_n)$. Hence, we have either $c(x_1, \dots, x_n)\leq 3(n-1)$ or $\pi(x_1, \dots, x_n)=0$ for all $\pi\in\Pi(m_1, \dots, m_n)$. This then demonstrates that,  similar to graph curvature, $-2\leq\kappa(E)\leq 1$.
	
	Specializing  Propositions~\ref{Prop:Equivalent_barycenter} and \ref{Prop:Duality} to the hypergraph setting, we now provide two equivalent formulations for hypergraph curvature.
	\begin{description}
		\item[Barycenter:]
		Although the minimization problem in \eqref{Def_Multi_Marginal} is a linear program, its complexity is exponential in $N$. However, it turns out \cite{Entropic_Regularized} that Wasserstein barycenter problem \eqref{Barycenter_Definition} can be solved quite efficiently. The equivalence between barycenter problem and multi-marginal optimal transport problem justifies to define the \textit{Wasserstein barycenter} of hyperedge $E$ as
		\begin{equation}
		\mathsf{bar}(E)\coloneqq \inf_{\nu\in \mathscr{P}(V)}\sum_{i=1}^nW(m_i, \nu).
		\end{equation}
		Following mutatis mutandis Proposition~\ref{Prop:Equivalent_barycenter}, we have  
		\begin{equation}\label{Equivalent_barycenter}
		\kappa(E)=1-\frac{\mathsf{bar}(E)}{n-1}.
		\end{equation}
		The term barycenter makes sense by recalling that for the Euclidean space the barycenter of points $\{\mathbf{x}_i\}_{i=1}^n$ is given by  $\argmin_{\mathbf{x}}\sum_{i=1}^n\|\mathbf{x}_i-\mathbf{x}\|^2$. 
		Consequently, $\mathsf{bar}(E)$ is the barycenter of points $\{m_i\}_{i=1}^n$ in the Wasserstein space (i.e., a metric space with Wasserstein distance).
		
		\item[Duality] Following mutatis mutandis Proposition~\ref{Prop:Duality}, we can write the following dual formula for $W(E)$ 
		\begin{equation}\label{Equivalence_Duality}
		\kappa(E)=1-\frac{1}{n-1}\sup_{(f_1, \dots, f_n)\in \mathcal{K}}\sum_{i=1}^n \mathbb{E}_{m_i}[f_i(X)] ,
		\end{equation}
		where $\mathbb{E}_\nu[\cdot]$ is the expectation operator with respect to measure $\nu$ and $\mathcal{K}$ is the set of all integrable real-valued functions on $V$ such that for any vector $(x_1, \dots, x_n)\in V^n$
		\begin{equation}\label{Constraint_Dual}
		\sum_{k=1}^nf_k(x_k)\leq c(x_1, \dots, x_n).
		\end{equation}
		
	\end{description}
	
	
	After defining hypergraph curvature, a natural question is whether or not this definition reduces to graph Ricci curvature \eqref{Def_Ollivier} if the hypergraph is indeed a simple graph, i.e., the cardinality of each hyperedge is two. We answer this question in affirmative by invoking the barycenter interpretation. If $H$ is in fact a graph and $E$ is a hyperedge (i.e., $n=2$ and thus $E=\{1,2\}$), then $\mathsf{bar}(E)$ equals either $m_1$ or $m_2$, because for any $\nu\in\mathscr{P}(V)$ the triangle inequality implies $W(m_1,\nu) + W(m_2, \nu)\geq W(m_1, m_2)$. Thus, hypergraph curvature coincides with \eqref{Def_Ollivier}.
	

	It is a well-known fact that Ollivier-Ricci curvature of edge $(x,y)$ depends heavily on the number of common neighbors of $x$ and $y$. Specifically, if $N(x)\cap N(y)=\emptyset$, then $\kappa(x,y)\leq 0$ (see e.g., \cite{Jost2014}). We now prove a similar result for hypergraphs. 
	
	\begin{theorem}\label{Thm:NecessaryCondition}
		For any hyperedge $E$ with cardinality $m$, we have
		$$W(E)\geq 1-\min_{\upsilon,\vartheta\in E}\sum_{F\in \E}\frac{|N(\upsilon)\cap N(\vartheta)\cap F|}{(d(F)-1)d_\upsilon}.$$
	\end{theorem}
	\begin{proof}
		Let again $E$ be denoted by $\{1, 2, \dots, n\}$. The proof relies on the dual formula of $W(E)$. In order to make use of \eqref{Equivalence_Duality}, we need to find a set of functions $\{f_j\}_{j=1}^n$ that satisfy the constraint \eqref{Constraint_Dual}. Fix two vertices $\upsilon$ and $\vartheta$ in $E$. Set $f_i\equiv 0$ for $i\in E\backslash\{\vartheta, \upsilon\}$ and suppose \begin{equation}\label{Proof_Necessary}
		f_\upsilon(r)+f_\vartheta(s)\leq d(r,s),
		\end{equation} for all pairs of vertices $(r,s)$. Then we have for any vector $(x_1, \dots, x_n)\in V^n$ 
		\begin{eqnarray*}
			\sum_{i=1}^nf_i(x_i)&=&f_\upsilon(x_\upsilon)+f_\vartheta(x_\vartheta)\\
			&\leq & d(x_\upsilon, x_\vartheta)= \min_{z\in V}\left[ d(x_\upsilon, z) + d(x_\vartheta, z)\right]\\
			&\leq& c(x_1, \dots, x_n).
		\end{eqnarray*}
		Hence, the constraints $f_i\equiv 0$ for $i\in E\backslash\{\vartheta, \upsilon\}$ and  \eqref{Proof_Necessary} are sufficient to satisfy \eqref{Constraint_Dual}. 
		
		Letting $\C_{\upsilon, \vartheta}$ denote the set of real-valued functions satisfying \eqref{Proof_Necessary}, we can write
		\begin{eqnarray*}
			W(E)&\geq& \sup_{f_\upsilon, f_\vartheta \in \C_{\upsilon, \vartheta}}\sum_{r\in V}f_\upsilon(r)m_\upsilon(r)+ \sum_{r\in V}f_\vartheta(r)m_\vartheta(r)\\
			&=& \sup_{f\in \A}\sum_{r\in V}f(r)\left[m_\upsilon(r)-m_\vartheta(r)\right],
		\end{eqnarray*}
		where the last equality is due to \cite[Theorem 1.14]{villani2003topics} and  $\A$ is the set of all real-valued $1$-Lipschitz functions $f$ on $V$, that is $$\A\coloneqq \{f:V\to \R: |f(r)-f(s)|\leq d(r,s)\}.$$ Now consider the following function
		$$
		f(r) =
		\begin{cases*}
		2, & if $r\in N(\upsilon)\backslash N(\vartheta)$, \\
		1,        & otherwise.
		\end{cases*}$$
		Clearly, $f\in \A$ and therefore, 
		\begin{eqnarray*}
			W(E)&\geq& \sum_{r\in V}f(r)\left[m_\upsilon(r)-m_\vartheta(r)\right]\\
			&=&2\sum_{r\in N(\upsilon)\backslash N(\vartheta)}m_\upsilon(r)+\sum_{r\in N(\upsilon)\cap N(\vartheta)}m_\upsilon(r)-1\\
			&=&1-m_\upsilon(N(\upsilon)\cap N(\vartheta))\\
			&=&1-\sum_{F\in \E}\frac{|N(\upsilon)\cap N(\vartheta)\cap F|}{(d(F)-1)d_\upsilon}.
		\end{eqnarray*}
	\end{proof} 
	In light of this theorem, we have 
	\begin{equation}\label{Bound_Curvature}
	\kappa(E)\leq \min_{\upsilon,\vartheta\in E}\sum_{F\in \E}\frac{|N(\upsilon)\cap N(\vartheta)\cap F|}{(d(F)-1)d_\upsilon}.
	\end{equation} 
	If $H$ happens to be a simple graph, then for every edge $E=(\upsilon,\vartheta)$, we have $m_{\upsilon}(N(\upsilon)\cap N(\vartheta))=\frac{\Delta}{d_\upsilon}$, where $\Delta$ is the number of triangles supported on edge $E$. Hence, the bound \eqref{Bound_Curvature} implies
	\begin{equation}\label{Bound_Graph}
	\kappa(\upsilon, \vartheta)\leq \frac{\Delta}{\max\{d_\upsilon, d_\vartheta\}},
	\end{equation}
	which  appears in \cite[Theorem 4]{Jost2014}.


	\begin{remark}[Curvature as a projection of Boltzmann]\label{Remark_Projection}
		Cuturi \cite{Entropic_Regularized} introduced an entropic regularized version of Wasserstein distance between two measures as 
		\begin{eqnarray}
		W_\eps(\mu, \nu)&\coloneqq& \inf_{\pi\in\Pi(\mu, \nu)}\sum_{(r,s)\in V\times V}\pi(r, s)d(r, s)-\eps H(\pi)\nonumber\\
		&=&\eps\inf_{\pi\in\Pi(\mu, \nu)} D(\pi\|\mathcal{K}_\eps)-\eps\log Z(\eps),\label{Def:Eontropic}
		\end{eqnarray}
		where $\eps>0$ is the regularization parameter, $H(\cdot)$ denotes the Shannon entropy function, $\mathcal{K}_\eps$ is the Boltzmann distribution defined as   $\mathcal{K}_\eps(x,y)\coloneqq \frac{e^{-d(x,y)/\eps}}{Z(\eps)}$ and $Z(\eps)\coloneqq \sum_{(x,y)\in V\times V}e^{-d(x,y)/\eps}$. It follows that $W(\mu, \nu)=\lim_{\eps\downarrow 0} W_\eps(\mu, \nu)$ \cite{Bregman_projection}. As mentioned in \cite{Entropic_Regularized}, despite the theoretically-guaranteed convergence,  the procedure cannot work beyond a graph-dependent value $\eps_0$ beyond which some entries of $\mathcal{K}_\eps$ are represented as zeroes in memory. Since $H(\cdot)$ is a strictly concave function, the optimization problem in \eqref{Def:Eontropic} has a unique solution $\pi^*$ which corresponds to the projection of $\mathcal{K}_\eps$ onto $\Pi(\mu, \nu)$.  
		Following this spirit, we can define the entropic regularized Wasserstein barycenter problem as 
		\begin{eqnarray}
		W_\eps(E)&\coloneqq& \inf_{\nu\in\mathscr{P}(V)}\sum_{i=1}^n W_\eps(m_i, \nu)\nonumber\\
		&=& \eps\inf_{\mathbb{P}\in \C} D\left(\mathbb{P}\|\mathcal{K}^{\otimes n}_\eps\right) -n\eps \log Z(\eps)\nonumber
		\end{eqnarray}
		where $\C\coloneqq \{\pi_1\otimes\pi_2\otimes\dots\otimes\pi_n\in \P^n(V\times V):~\pi_i\mathbf{1}=m_i,~\text{and}~\exists \eta\in\mathscr{P}(V) ~\text{s.t.}~\pi_i^T\mathbf{1}=\eta\}$. Hence, $W_\eps(E)$ is the unique projection of $\mathcal{K}_\eps^{\otimes n}$ onto $\C$. Assuming $\eps$ is sufficiently small, it follows that $W(m_1, \dots, m_n)$ is an approximation of this projection.
		
	\end{remark}
	\section{Complete uniform hypergraphs}
	\label{sec:Uniform_Hyper}
	Graph Ricci curvature turns out to have a simple formula for complete graphs $K_N$. In particular, it is shown \cite{Jost2014} that $\kappa(x,y)=\frac{N-2}{N-1}$ for any edges $(x,y)$ in $K_N$. In this section, we show that complete uniform hypergraphs with $N$ vertices have the same curvature as  $K_N$. 
	\begin{definition}
		A hypergraph $H^n_N=(V, \mathcal{E})$ is called \textit{complete $n$-uniform} for $n\leq |V|=N$ if $\mathcal{E}=[V]_n$, where $[V]_n$ is the collection of all $n$-subsets of $V$.
	\end{definition}
	Notice that $H^2_N=K_N$. Recall that $V=\{1, 2, \dots, N\}$. It follows that in $H^n_N$, we have $d(j)={N-1 \choose n-1}$, $j\in V$ and  $d(E)=n$ for each $E\in \E$. Furthermore, any pair of $(i, j)\in E$ is contained in ${N-2 \choose n-2}$ many hyperedges. Therefore, according to \eqref{Def_RandomWalk}, the random walk started at $j\in V$ of $H^n_N$ is
	\begin{equation}\label{RandomWalk_Complete}
	m_j(r)=\frac{1_{\{r\neq j\}}}{N-1}\qquad \forall r\in V.
	\end{equation}
	\begin{lemma}\label{Lemma_Complete_Unifom}
		For every hyperedge $E$ of $H^n_N$, we have
		$$W(E)=\frac{n-1}{N-1}.$$
		In particular, 
		$$\kappa(E)=\frac{N-2}{N-1}.$$ 
	\end{lemma}
	In light of this lemma, the curvature of complete $n$-uniform hypergraphs is independent of $n$. This is, in fact, a result of the normalization in Definition~\ref{Def_Hypergraph_Curva}. 
	\begin{proof} 
		Again let $E=\{1, 2, \dots, n\}$.  Notice that for $H^n_M$, we have $d(r,s)=1$ for any distinct pair of vertices $r, s\in V$. Thus, we can write for a measure $\nu\in \mathscr{P}(V)$  
		\begin{eqnarray*}
			\sum_{i=1}^n W(m_1, \nu) &=& \min_{\pi^{(1)}\in\Pi(m_1, \nu)}\sum_{r=1}^{N}\sum_{s=1}^{N}d(r,s)\pi^{(1)}(r,s)+ \dots   + \min_{\pi^{(n)}\in\Pi(m_n, \nu)}\sum_{r=1}^{N}\sum_{s=1}^{N}d(r,s)\pi^{(n)}(r,s) 
			\\
			&=&\min_{\pi^{(1)}\in\Pi(m_1, \nu)}\left[1-\sum_{r=1}^{N}\pi^{(1)}(r,r)\right] +\dots +   \min_{\pi^{(n)}\in\Pi(m_n, \nu)}\left[1-\sum_{r=1}^{N}\pi^{(n)}(r,r)\right] \\
			&=&\frac{1}{2}\sum_{i=1}^n\|m_i-\nu\|,
		\end{eqnarray*}
		where the last equality follows from \cite[Exercise 1.17]{villani2003topics} and for two vectors $a$ and $b$, $\|a-b\|\coloneqq \sum_{i}|a(i)-b(i)|$. Thus, we obtain 
		$$\mathsf{bar}(E)=\inf_{\nu\in\mathscr{P}(V)}\frac{1}{2}\sum_{i=1}^n\|m_i-\nu\|=\frac{1}{2}\sum_{i=1}^n\|m_i-\bar{m}\|,$$
		where $\bar{m}=\frac{1}{n}\sum_{i=1}^nm_i$.        
		Recalling that $m_i(r)=\frac{1_{\{r\neq i\}}}{N-1}$, we obtain
		$$W(E)=\frac{n-1}{N-1}.$$
	\end{proof}
	\section{Examples}
	\label{Sec:Example}
	In this section, we focus on the computation of hypergraph curvature in two examples to illustrate the differences and similarities with the graph curvature. The first example is a natural generalization of infinite path $P_n$; a simple graph that has $n$ vertices with $n-2$ vertices of degree $2$ and the other two of degree one. 
	
	\begin{example}
		It is shown \cite{Exact_RicciResults} that $P_n$ is one of the few graphs (among graphs with girth at least $5$) with constant zero curvature.  We now demonstrate in the following theorem that a similar statement does not hold for the \textit{hyperpath}; a hypergraph whose vertices have degree at most 2. For the ease of presentation, we assume that any two intersecting hyperedges have exactly one common vertex; see Fig.~\ref{fig:Hypergraph}.
		\begin{figure}
			\centering
			\begin{tikzpicture}[scale=0.7]
			\draw[thick, black!40!red, rotate=0] (0,0) ellipse (38pt and 25pt);
			\draw[thick, black!30!green,rotate=0] (3.8,0) ellipse (38pt and 25pt);
			\draw[thick, black!40!blue,rotate=0] (1.8,0) ellipse (38pt and 25pt);
			\draw[thick, black!30!brown,rotate=0] (5.8,0) ellipse (38pt and 25pt);
			\draw[thick, black!20!gray,rotate=0] (7.8,0) ellipse (38pt and 25pt);
			\draw[thick, black,rotate=40] (4,-1.6) ellipse (30pt and 20pt);
			\draw[thick, red!40!black,rotate=-40] (3.3,4.45) ellipse (30pt and 20pt);
			\fill (-0.4,-0.4) circle (0.06) node [right]{};
			\fill (-0.9,0) circle (0.06) node [right]{};
			\fill (-0.6,0.4) circle (0.06) node [right]{};
			\fill (-0.2,0.1) circle (0.06) node [right]{};
			\fill (0.2, 0.5) circle (0.06) node [right] {};
			\fill (0.4, -0.4) circle (0.06) node [right] {};

			\fill (1.6,-0.4) circle (0.06) node [right]{};
			\fill (2.4,0.4) circle (0.06) node [right]{};
			\fill (1.4,0.4) circle (0.06) node [right]{};
			\fill (1.8,0.1) circle (0.06) node [right]{};
			\fill (2, 0.5) circle (0.06) node [right] {};
			\fill (2.4, -0.4) circle (0.06) node [right] {};

			\fill (3.6,-0.4) circle (0.06) node [right]{};
			\fill (3.4,-0.5) circle (0.06) node [right]{};
			\fill (4,0.1) circle (0.06) node [right]{};
			\fill (3.8,0.1) circle (0.06) node [right]{};
			\fill (3.2, 0.2) circle (0.06) node [right] {};
			\fill (4.4, -0.4) circle (0.06) node [right] {};
			\fill (3.8, 0.6) circle (0.06) node [right] {};
			\fill (4, -0.4) circle (0.06) node [right] {};
			
			\fill (5.6,-0.4) circle (0.06) node [right]{};
			\fill (5.8, 0.2) circle (0.06) node [right]{};
			\fill (5.4, 0.1) circle (0.06) node [right]{};
			\fill (5.7, 0.65) circle (0.06) node [right]{};
			\fill (5.6, -0.1) circle (0.06) node [right]{};
			\fill (6.2, 0.3) circle (0.06) node [right] {};
			\fill (6.4, -0.4) circle (0.06) node [right] {};
			
			\fill (5.4, 1.65) circle (0.06) node [right]{};
			
			\fill (8,-0.4) circle (0.06) node [right]{};
			\fill (7.8, 0.2) circle (0.06) node [right]{};
			\fill (7.4, 0.4) circle (0.06) node [right]{};
			\fill (7.6, -0.1) circle (0.06) node [right]{};
			\fill (8.3, 0.5) circle (0.06) node [right] {};
			\fill (8.8, -0.4) circle (0.06) node [right] {};
			
			\fill (3.8,1.8) circle (0.06) node [right]{};
			\fill (4.3, 2.1) circle (0.06) node [right]{};
			\fill (4.6, 1.6) circle (0.06) node [right]{};
			\fill (4.1, 1) circle (0.06) node [right]{};

			\fill (0.9,0) circle (0.06) node [right]{};
			\fill (2.9,0) circle (0.06) node [right]{};
			\fill (4.9,0) circle (0.06) node [right]{};
			\fill (6.9,0) circle (0.06) node [right]{};

			\end{tikzpicture}
			\caption{A hyperpath with 42 vertices and 7 hyperedges. For the green hyperdge, we have $m=10$ and $\beta= 7$. Note that, unlike path graph, hyperpath might have cycles.}
			\label{fig:Hyperpath}
		\end{figure}
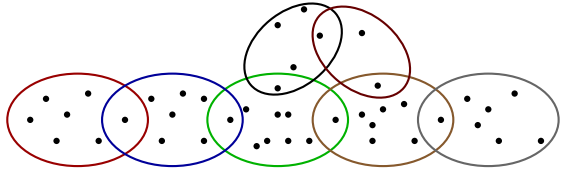
		\begin{theorem}\label{Thm_HyperPath} 
			For any hyperedge $E$ of a hyperpath described above, we have  	$$W(E) \leq \beta \frac{2n-3}{2(n-1)}+(n-\beta-1)\frac{3n-4}{2(n-1)},$$
			where $n\coloneqq d(E)\geq 3$ and $\beta\coloneqq |\{v\in E:~\nexists R\in \mathcal{E}, ~v\in R\}|<n$. In particular, the curvature of hyperedge $E$ is bounded from below by
			\begin{equation*}
			\kappa(E)\geq \frac{-(n-\beta-2)}{2(n-1)}.
			\end{equation*}
		\end{theorem}	
		\begin{proof}[Proof Sketch]
			The proof relies on the simple observation that for the hyperedge $E=\{1, 2, \dots, n\}$, we have  $W(E)\leq \sum_{i=1}^n W(m_i, \nu)$ for any probability measure $\nu\in \mathscr{P}(V)$. In particular, \begin{equation}\label{Wasser0}
			W(E)\leq \sum_{i=2}^nW(m_i, m_1).
			\end{equation}
			
			Assuming that $\{1,2, \dots, n-\beta\}$ are vertices in $E$ that are shared with other hyperedges and $\{n-\beta+1, \dots, n\}$ are isolated vertices inside $E$, we can employ a tedious (yet standard) linear-programming argument to show that 
			\begin{equation}\label{Wasser1}
			W(m_1, m_j)=\frac{3n-4}{2(n-1)}, \qquad \forall j\in \{2, \dots, n-\beta\},
			\end{equation}
			and 
			\begin{equation}\label{Wasser2}
			W(m_1, m_k)=\frac{2n-3}{2(n-1)}, \qquad \forall k\in \{n-\beta+1, \dots, n\}.
			\end{equation}
			Plugging \eqref{Wasser1} and \eqref{Wasser2} into \eqref{Wasser0}, we obtain the result. 
		\end{proof}
		In light of this result, we have $\kappa(E)>0$ if only one vertex of $E$ is shared, i.e., $\beta>n-2$.  In other words, the leaves of a hyperpath have positive curvature which is  different from graph curvature in that in simple graphs (with girth at least 5) each edge connecting to a leaf has zero curvature, see \cite[Theorem 3.3]{Exact_RicciResults}.  
	\end{example}
	
	\begin{example}\label{Example_Toy}
		As a toy example, consider the hypergraph $H=(V, \E)$ with $V=\{1, 2, \dots, 13\}$ and $\E=\{E_1, E_2, E_3, E_4\}$, where $E_1=\{1, 2, 3\}$, $E_2=\{2, 4,\dots, 7\}$, $E_3=\{6, \dots, 11\}$, and $E_4=\{7, 11, 12, 13\}$, as illustrated in Fig.~\ref{fig:Hypergraph}. Using \eqref{Def_RandomWalk}, we can compute the probability measures associated to each hyperedge. For instance, the random walk started at vertex $2$ is $\mu_2=[0.25, 0, 0.25, 0.125, 0.125, 0.125, 0.125, 0, \dots, 0]$. Since there are only 13 vertices, we can solve the optimization problem \eqref{Def_Multi_Marginal} (as a linear program) for each hyperedge. Solving this optimization problem, we obtain $W(E_1)=1$, $W(E_2)=2.38$, $W(E_3)=2.08$, and $W(E_4)=1.45$. Consequently, it follows that $\kappa(E_1)=0.5$, $\kappa(E_2)=0.4$, $\kappa(E_3)=0.58$, and $\kappa(E_4)=0.52$.
		
		Informally speaking, the hyperedge with the lowest curvature is the bridge connecting two components of hypergraphs. This is similar to the intuition of graph Ricci curvature, as experimentally observed in \cite{Curvature_Internet_Topolo}, that edges with negative curvature are locally shortcuts of two component of graph.   
	\end{example}
	
	\newcommand{\boundellipse}[3]
	{(#1) ellipse (#2 and #3)
	}
	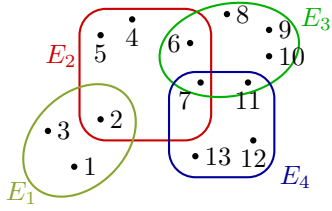
\begin{figure}
		\centering
		\begin{tikzpicture}[scale=0.7]
		\node (v5) at (0.4,2) {};
		\node (v10) at (3.6,1.6) {};
		\node (v4) at (1,2.3) {};
		\node (v2) at (0.4,0.4) {};
		\node (v13) at (2.2,-0.3) {};
		\node (v12) at (3.3,0) {};
		\node (v11) at (3.2,1.1) {};
		\node (v9) at (3.6,2.1) {};
		\node (v1) at (-0.1,-0.5) {};
		\node (v3) at (-0.6,0.18) {};
		\node (v8) at (2.8,2.4) {};
		\node (v7) at (2.3,1.1) {};
		\node (v6) at (2.1,1.85) {};
		\draw [thick, black!20!red, rounded corners=10, draw]
		(0,0) --
		++(2.5,0) --
		++(0,2.5) --
		++(-2.5,0) --
		cycle
		{};
		\draw[thick, brown!70!green, rotate=-45] (0,0) ellipse (25pt and 35pt);
		\draw[thick, black!30!green,rotate=10] (3.1,1.2) ellipse (38pt and 25pt);
		\draw [thick, black!40!blue, rounded corners=10, draw]
		(1.7,-0.7) --
		++(2,0) --
		++(0,2) --
		++(-2,0) --
		cycle
		{};
		\fill (v13) circle (0.06) node [right]{$13$};
		\fill (v12) circle (0.06) node [below]{$12$};
		\fill (v11) circle (0.06) node [below]{$11$};
		\fill (v8) circle (0.06) node [right] {$8$};
		\fill (v3) circle (0.06) node [right] {$3$};
		\fill (v2) circle (0.06) node [right] {$2$};
		\fill (v4) circle (0.06) node [below] {$4$};
		\fill (v5) circle (0.06) node [below] {$5$};
		\fill (v6) circle (0.06) node [left] {$6$};
		\fill (v7) circle (0.06) node [below left] {$7$};
		\fill (v9) circle (0.06) node [right] {$9$};
		\fill (v1) circle (0.06) node [right] {$1$};
		\fill (v10) circle (0.06) node [right] {$10$};
		\node (E1) at (-1.1,-1) {\textcolor{brown!70!green}{$E_1$}};
		\node (E2) at (-0.35,1.6) {\textcolor{black!20!red}{$E_2$}};
		\node (E3) at (4.5,2.3) {\textcolor{black!30!green}{$E_3$}};
		\node (E4) at (4.1,-0.7) {\textcolor{black!40!blue}{$E_4$}};
		\end{tikzpicture}
		\caption{The hypergraph studied in Example~\ref{Example_Toy}.}
		\label{fig:Hypergraph}
	\end{figure}	

	\section{Coarse Scalar Curvature in a Riemannian Manifold Setting}
	\label{sec:coarse-scalar-curvature-manifold}
	
	In this section, we establish a consistency result for the definition of coarse scalar curvature \eqref{eq:coarse-scalar-curvature} under a Riemannian manifold setting. We will assume the hypergraphs of interest stem from a natural geometric model encoding adjacency relations among discrete points randomly sampled from a Riemannian manifold.
	
	\begin{definition}[$\eps$-Neighborhood Hypergraphs]
		Let $\left( M,d_M \right)$ be a complete metric space. For any $\eps>0$, an \emph{$\eps$-neighborhood hypergraph $H=\left( V,\E \right)$ supported on $M$} consists of the following data:
		\begin{itemize}
			\item A finite set of vertices $V=\left\{ v_i\in M,\,1\leq i\leq N \right\}$;
			\item A set of hyperedges $\E$, where each $E\in\E$ is a subset of $V$ satisfying
			\begin{equation*}
			\mathrm{diam}\left( E \right):=\sup_{u,v\in E} d_M \left( u,v \right) < 2\eps;
			\end{equation*}
			\item Every subset $U$ of vertices with $\mathrm{diam}\left( U \right)<2\eps$ corresponds to a hyperedge $E_U\in \E$.
		\end{itemize}
		We will denote $\left| E \right|$ for the number of vertices joined by hyperedge $E$; $\left| V \right|$ and $\left| \E \right|$ will be used to denote the number of vertices and hyperedges in $H$, respectively.
	\end{definition}
	The definition of $\eps$-neighborhood hypergraphs assumes that the spatial proximity of the vertices is faithfully reflected in the connectivity: the diameter of every hyperedge --- understood as a discrete subset $E\subset V$ --- is bounded from above by $2\eps$. 
	
	In the rest of this section, we will assume $M$ is an orientable $d$-dimensional complete Riemannian manifold with finite volume, and the distance function $d_M:M\times M\rightarrow \left[ 0,\infty \right)$ is induced from the Riemannian length structure. As often encountered in the manifold learning setup (see e.g. \cite{LapEigMaps2003} and \cite{CoifmanLafon2006}), we assume the vertices of $H$ are sampled i.i.d. uniformly on $M$ with respect to the standard volume measure $\mathrm{dvol}/\mathrm{vol}\left( M \right)$. Our goal in this section is to establish conditions under which the coarse scalar curvature \eqref{eq:coarse-scalar-curvature} converges to the scalar curvature on $M$ in the limit $\left| V \right|\rightarrow\infty$ and $\eps\rightarrow 0$. We begin our discussion by summarizing some geometric and statistical properties of the Riemannian medians that we will need in the following two lemmas.
	\begin{lemma}
		\label{lem:riemannian-median}
		Let $M$ be an orientable $d$-dimensional complete Riemannian manifold with
		\begin{itemize}
			\item Finite volume $\mathrm{vol}\left( M \right)<\infty$;
			\item Finite, positive injectivity radius $0<\mathrm{Inj}\left( M \right)<\infty$;
			\item Uniformly bounded sectional curvature $K_x \left( u,v \right)<D$ for $D\in\mathbb{R}$ and all $x\in M$, $u,v\in T_xM$.
		\end{itemize}
		Assume $U=\left\{ x_1,\cdots,x_n \right\}$ is a finite discrete subset of $M$ contained in a geodesic ball $B \left( x,\eps \right)$ of sufficiently small radius $\eps>0$ satisfying
		\begin{equation}
		\label{eq:epsilon-assumption}
		2\eps <
		\begin{cases}
		\displaystyle \min \left\{ D^{-1/2}\pi/4, \mathrm{Inj}\left( M \right)/2 \right\} & \textrm{if $D>0$}\\
		\displaystyle \mathrm{Inj}\left( M \right)/2 & \textrm{otherwise}
		\end{cases}
		\end{equation}
		and $U$ is not totally contained in any geodesic on $M$. Then there exists a unique $\bar{x}\in M$ minimizing the moment function
		\begin{equation*}
		M\ni z\mapsto \sum_{i=1}^nd_M \left( z, x_i\right)\in \left[ 0,\infty \right),
		\end{equation*}
		and $\bar{x}$ lies in the smallest closed, geodesically convex subset of $B \left( x,\eps \right)$ containing $U$.
	\end{lemma}
	\begin{proof}
		The bound on $\eps$ and the assumption that $U$ is not totally contained in any geodesic on $M$ together ensures the uniqueness of $\bar{x}$, according to \cite[Theorem 3.1]{Yang2010}; the existence can be found e.g. in \cite{FVJ2009geometric} or \cite{ABY2013}. The last assertion about the location of $\bar{x}$ is a consequence of \cite[Proposition 2.4]{Yang2010}.
	\end{proof}
	The minimizer $\bar{x}$ is known as the \emph{Riemannian median}. Note that the diameter bound $2\eps$ only depends on the injectivity radius when the Riemannian manifold has negative sectional curvatures; this is of particular interest to us since many real world networks demonstrate ``negatively curved'' behaviors, see e.g., \cite{Curvature_Saniee,Curvature_MObasheri,Curvature_Internet_Topolo}.
	
	\begin{lemma}
		\label{lem:riemannian-median-consistency}
		Assume $M$, $\eps>0$ satisfy the same assumptions as in Lemma~\ref{lem:riemannian-median}. Let $\left\{ X_n \right\}_{n=1}^{\infty}$ be a sequence of i.i.d. random points uniformly distributed in a geodesic ball $B \left( x,\eps \right)$ of radius $\eps>0$ centering at $x\in M$, and denote $\mu_n$ for the Riemannian median of $X_1,\cdots,X_n$ for all $n\in\mathbb{N}$. Then $\mu_n\rightarrow \bar{x}$ a.s. as $n\rightarrow\infty$, where $\bar{x}$ is the unique Riemannian median of the uniform distribution on $B \left( x,\eps \right)$.
	\end{lemma}
	\begin{proof}
		The existence and uniqueness of $\mu_n$ and $\bar{x}$ follows from Lemma~\ref{lem:riemannian-median}; note here that $\bar{x}$ need not coincide with $x$ in general, by the characterization of Riemannian medians established in \cite[Proposition 2.1 and Theorem 2.2]{Yang2010}. The almost sure convergence has been established in \cite[Corollary 4.1]{ABY2013}.
	\end{proof}
	
	We now turn to investigating the coarse scalar curvature of hypergraphs generated from a geometric probabilistic model: saturated $\eps$-neighborhood hypergraphs supported on a Riemannian manifold $M$, with vertices uniformly distributed on $M$. We consider the same random walk on the Riemannian manifold $M$ as in \cite{Ollivier1}, namely, the one given in \eqref{eq:random-walk-riemannian-manifold}: for any $x\in M$ and $\eps>0$,
	\begin{equation}
	\label{eq:random-walk-riemannian-manifold-epsilon}
	\text{d}m_x^\eps\coloneqq \frac{1_{B\left(x, \eps\right)}}{\text{vol}(B(x, \eps))}\text{d}\text{vol},
	\end{equation}
	which is essentially the standard volume measure on $M$ restricted and renormalized on $B \left( x,\eps \right)$. For simplicity of statement, let us denote $\hat{x}_E$ for the Riemannian median --- when it exists and is unique --- of the vertices connected by a hyperedge $E$ in a saturated hypergraph $G=\left( V,\E \right)$.
	
	\begin{theorem}
		\label{thm:coarse-scalar-curvature-riem-manifold}
		Let $M$, $\eps>0$ be as assumed in Lemma~\ref{lem:riemannian-median}. Let $\left\{ v_i \right\}_{i=1}^{\infty}$ be a sequence of i.i.d. random points sampled uniformly on $M$ with respect to the standard Riemannian volume measure. For any $N\in\mathbb{N}$, let $V_N:=\left\{ v_i \right\}_{i=1}^N$ and $H_N=\left( V_N,\E_N \right)$ be an $\eps$-neighborhood hypergraph supported on $M$. If there exists a hyperedge $E_N\in\E_N$ for each $H_n$ such that
		\begin{itemize}
			\item $E_N\subsetneqq E_{N+1}$ for all sufficiently large $N\in\mathbb{N}$
			\item there exists $\bar{x}\in M$ such that $\hat{x}_{E_N}\rightarrow\bar{x}$ as $n\rightarrow\infty$
			\item $E_N\subset B \left( \bar{x},\eps \right)$ for all sufficiently large $N\in\mathbb{N}$
		\end{itemize}
		then the coarse scalar curvature of hyperedge $E_N$ satisfies, for all sufficiently small $\eps>0$, 
		\begin{equation}
		\label{eq:coarse-to-smooth-scalar}
		\limsup_{N\rightarrow\infty}\kappa \left( E_N \right)\geq \frac{\eps^2}{2 \left( d+2 \right)}\mathrm{Scal}\left( \bar{x} \right)+O \left( \eps^3 \right).
		\end{equation}
	\end{theorem}
	\begin{proof}
		Since curvature is a local quantity, we may assume without loss of generality that $M$ is connected and even compact. Let $\eps>0$ be sufficiently small such that the geodesic ball $B \left( \bar{x},\eps \right)$ is geodesically convex neighborhood of $\bar{x}$ (c.f. \cite[Proposition 4.2]{doCarmo1992RG}). To ease notations, write $\ell_N:=\left| E_N \right|$ and denote $x_1,\cdots,x_{\ell_N}\in V_N$ for the vertices of $H_N$ connected by the hyperedge $E_N\in\E$. By Proposition~\ref{Prop:Equivalent_barycenter} and the definition of Riemannian median, we have
		\begin{equation}\label{eq:upperbound-limit}
		\frac{W_1 \left( E_N \right)}{c \left( x_1,\cdots,x_{\ell_n} \right)}\leq \frac{\sum_{i=1}^{\ell_N}W_1 \left( m_{x_i},m_{\bar{x}} \right)}{\sum_{i=1}^{\ell_N}d_M \left( x_i,\hat{x}_{E_N} \right)}
		\xrightarrow{\ell_N\rightarrow\infty}\frac{\displaystyle \int_{B \left( \bar{x},\eps \right)}W_1 \left( m_y,m_{\bar{x}} \right)\frac{\mathrm{dvol} \left( y \right)}{\mathrm{vol}\left( B \left( \bar{x},\eps \right) \right)}}{\displaystyle \int_{B \left( \bar{x},\eps \right)}d_M \left( y,\bar{x} \right)\frac{\mathrm{dvol} \left( y \right)}{\mathrm{vol}\left( B \left( \bar{x},\eps \right) \right)}}
		\end{equation}
		where the limit follows from the law of large number and the Lipschitz continuity of the function (see e.g. \cite[\textsection 2]{Yang2010})
		\begin{equation*}
		x\mapsto \int_{B \left( x,\eps \right)}d_M \left( x,y \right)\,\mathrm{dvol}\left( y \right).
		\end{equation*}
		We know from \cite[Example 7]{Ollivier1} that, for any $y\in B \left( \bar{x},\eps \right)$,
		\begin{equation*}
		W_1 \left( m_y, m_{\bar{x}} \right)\\
		=\left( 1-\frac{\eps^2}{2 \left( d+2 \right)}\mathrm{Ric}\left( v_{\bar{x},y}, v_{\bar{x},y}\right)+O \left( \eps^3 \right) \right)d_M \left( \bar{x},y \right)
		\end{equation*}
		where $v_{\bar{x},y}$ is a unit tangent vector in $T_{\bar{x}}M$ such that\
		\begin{equation*}
		\exp_{\bar{x}}\left( d_M \left( \bar{x},y \right)v_{\bar{x},y} \right)=y.
		\end{equation*}
		It follows that
		\begin{equation*}
		\begin{aligned}
		\int_{B \left( \bar{x},\eps \right)}W_1 \left( m_y,m_{\bar{x}} \right)\,\mathrm{dvol} \left( y \right)&=\int_{B \left( \bar{x},\eps \right)}d_M \left( \bar{x},y \right)\mathrm{dvol}\left( y \right)-\frac{\eps^2}{2 \left( d+2 \right)}\int_{B \left( \bar{x},\eps \right)}\mathrm{Ric}\left( v_{\bar{x},y}, v_{\bar{x},y}\right)d_M \left( \bar{x},y \right)\mathrm{dvol}\left( y \right)\\
		&~~~+O \left( \eps^3\int_{B \left( \bar{x},\eps \right)}d_M \left( \bar{x},y \right)\mathrm{dvol}\left( y \right) \right)
		\end{aligned}
		\end{equation*}
		which, when plugged back into \eqref{eq:upperbound-limit}, gives
		\begin{equation}
		\limsup_{N\rightarrow\infty}\frac{W_1 \left( E_N \right)}{c \left( x_1,\cdots,x_{\ell_N} \right)}\leq1 +O \left( \eps^3 \right)
		-\frac{\eps^2}{2 \left( d+2 \right)}\frac{\displaystyle \int_{B \left( \bar{x},\eps \right)}\mathrm{Ric}\left( v_{\bar{x},y}, v_{\bar{x},y}\right)d_M \left( \bar{x},y \right)\frac{\mathrm{dvol}\left( y \right)}{\mathrm{vol}\left( B \left( \bar{x},\eps \right) \right)}}{\displaystyle \int_{B \left( \bar{x},\eps \right)}d_M \left( \bar{x},y \right)\frac{\mathrm{dvol}\left( y \right)}{\mathrm{vol}\left( B \left( \bar{x},\eps \right) \right)}}.\label{eq:limsup-intermediate}
		\end{equation}
		A straightforward calculation using geodesic normal coordinates reveals
		\begin{equation}
		\label{eq:asymp-expan-moment}
		\int_{B \left( x,\eps \right)}d_M \left( x,y \right)\,\frac{\mathrm{dvol}_M \left( y \right)}{\mathrm{vol}\left( B \left( x,\eps \right) \right)}
		=\eps\left[\frac{d}{d+1}-\frac{\eps^2}{3} \frac{\mathrm{Scal}\left( x \right)}{\left( d+1 \right)\left( d+2 \right)\left( d+3 \right)}+O \left( \eps^3 \right)\right]
		\end{equation}
		and
		\begin{equation}
		\int_{B \left( x,\eps \right)}\mathrm{Ric}\left( v_{x,y}, v_{x,y}\right)d_M \left( x,y \right)\,\frac{\mathrm{dvol}_M \left( y \right)}{\mathrm{vol}\left( B \left( x,\eps \right) \right)}
		=\eps \Bigg[ \frac{\mathrm{Scal}\left( x \right)}{d+1}+\frac{\eps^2}{3\left( d+2 \right)\left( d+3 \right)}\left( \frac{\left| \mathrm{Scal}\left( x \right) \right|^2}{d+1}-\left\| R \left( x \right) \right\|^2 \right)
		+O \left( \eps^3 \right) \Bigg] \label{eq:asymp-expan-ricci-moment}
		\end{equation}
		where $\left\| R \left( x \right) \right\|=\sum_{i=1}^d\left|\mathrm{Ric}\left( e_i,e_i \right)\right|^2$ for an arbitrary orthonormal basis $e_1,\cdots,e_d$ of $T_{\bar{x}}M$. Plugging \eqref{eq:asymp-expan-moment} and \eqref{eq:asymp-expan-ricci-moment} back into \eqref{eq:limsup-intermediate} to conclude that
		\begin{equation}
		\label{eq:upperbound}
		\limsup_{\ell_N\rightarrow\infty}\frac{W_1 \left( E_N \right)}{c \left( x_1,\cdots,x_{\ell_N} \right)}\leq 1-\frac{\eps^2}{2 \left( d+2 \right)}\mathrm{Scal}\left( \bar{x} \right)+O \left( \eps^3 \right).
		\end{equation}

	\end{proof}
	
	Theorem~\ref{thm:coarse-scalar-curvature-riem-manifold} indicates that coarse scalar curvature asymptotically upper bounds the scalar curvature of the Riemannian manifold, when the hypergraph is constructed from uniformly sampling the manifold in a natural way. This justifies the nomenclature of "scalar curvature" in our definition. We conjecture that coarse scalar curvature also asymptotically lower bounds the scalar curvature in the same setting, but will have to leave that for future work.

	\section{Conclusion and future work}
	\label{sec:Conclusion}
	In this paper, we propose a novel definition of curvature for hypergraphs by generalizing coarse Ricci curvature to coarse scalar curvature through multi-marginal optimal transport. Our definition is shown to be consistent with graph curvature in that (i) it reduces to graph curvature if the hypergraph of interest is indeed a simple graph, and (ii) it shares several properties with graph curvature. In particular, it is experimentally observed that, analogous to graph curvature, hypergraph curvature can be used to determine the bridge between components in the network. We are currently computing hypergraph curvature in real-world hypergraph networks (in particular co-authorship or cellular networks)  to observe this centrality property of hypergraph curvature. We are also applying hypergraph curvature to characterize dynamic effects in large dynamic network (in particular, financial network). Intuitively, hypergraph curvature provides a computational method for detecting changes in dynamic networks,  characterizing fast evolving network components, as well as identifying stable network region. On the theoretical side, we are interested in gaining better understandings of our hypergraph curvature with deeper insights from differential geometry.


	\bibliographystyle{IEEEtran}
	\bibliography{bibliography}

\end{document}